\documentclass[11pt]{article}
\usepackage[utf8]{inputenc}
\usepackage[a4paper, total={6in, 8in}]{geometry}
\usepackage{amsmath}
\usepackage{amsfonts}
\usepackage{comment}
\usepackage{amsthm}
\usepackage{algpseudocode}
\usepackage{amssymb}

\usepackage{xcolor}
\usepackage{hyperref}
\hypersetup{
    colorlinks=true,
    linkcolor=brown,
    filecolor=brown,  
    citecolor=brown
}

\theoremstyle{definition}
\newtheorem{definition}{Definition}[section]
\newtheorem{theorem}{Theorem}[section]
\newtheorem{algorithm}{Algorithm}[section]

\def\zsi{z^{(i)}}
\def\zsj{z^{(j)}}
\def\xsi{x^{(i)}}
\def\xsj{x^{(j)}}

\def\ysi{y^{(i)}}
\def\ysj{y^{(j)}}
\def\rsi{\rho^{(i)}}
\def\rsj{\rho^{(j)}}
\def\csi{\sigma^{(i)}}

\renewcommand{\title}[1]{\vbox{\center\LARGE{#1}}\vspace{5mm}}
\renewcommand{\author}[1]{\vbox{\center#1}\vspace{5mm}}

\newcommand{\address}[1]{\vbox{\center\em#1}}

\begin{document}

\title{Quantum Semi-Supervised Learning with Quantum Supremacy}

\begin{center}
\vspace{5mm}
Zhou Shangnan$^{a}$\footnote{\tt snzhou@stanford.edu.}
\end{center}

\address{${}^a$Stanford Institute for Theoretical Physics, \\
Stanford University, Stanford, CA 94305, USA}

\begin{abstract}
    Quantum machine learning promises to efficiently solve important problems. There are two persistent challenges in classical machine learning: the lack of labeled data, and the limit of computational power. We propose a novel framework that resolves both issues: quantum semi-supervised learning. Moreover, we provide a protocol in systematically designing quantum machine learning algorithms with quantum supremacy, which can be extended beyond quantum semi-supervised learning. In the meantime, we show that naive quantum matrix product estimation algorithm outperforms the best known classical matrix multiplication algorithm.
    We showcase two concrete quantum semi-supervised learning algorithms: a quantum self-training algorithm named the propagating nearest-neighbor classifier, and the quantum semi-supervised K-means clustering algorithm. By doing time complexity analysis, we conclude that they indeed possess quantum supremacy.
\end{abstract}

\section{Introduction}

Machine learning has made many seemingly impossible tasks possible: from visual and speech recognition, effective web search, to study of human genomics \cite{lecun2015deep,alphafold}. However, there are several long-standing bottlenecks in the field of machine learning, which slows down its pace in conquering more fields of science and technology. Two major challenges are the lack of labeled data, and the limit of computational power. In this paper, we propose a framework of quantum semi-supervised learning, which can overcome both difficulties at the same time. 

Semi-supervised learning \cite{zhu2009introduction} combines a small amount of labeled data with a large amount of unlabeled data during training, which tackles the common issue of the lack of labeled data. Quantum computation \cite{NielsenChuang} redefines the way computers create and manipulate information. Many quantum algorithms \cite{deutsch1992rapid, shor1994algorithms, grover1996fast, bernstein1997quantum, simon1997power, lloyd2013quantum, arute2019quantum} have been demonstrated to possess quantum supremacy, which means that they can execute the same task substantially faster than their classical counterparts. Quantum semi-supervised learning combines the advantages of both semi-supervised learning and quantum computation, therefore represents the future of machine learning and quantum physics.

Since semi-supervised learning handles both labeled and unlabeled data, in the limiting case when we only have labeled or unlabeled data, we get back to supervised or unsupervised learning. Hence, in some way, semi-supervised learning is more generic, and its algorithms can be used in supervised or unsupervised learning with small modifications. Therefore, many discussions in this paper also apply to quantum supervised or unsupervised learning.

In section \ref{framework}, we propose a general framework of quantum semi-supervised learning. In section \ref{supremacy}, we provide a generic protocol for designing quantum machine learning algorithms with quantum supremacy, which can be extended beyond quantum semi-supervised learning. The recipe can also be used to realize quantum supremacy in deep quantum neutral networks. One feature of our time complexity analysis is that we give a clear separation between memory access time complexity and algorithmic time complexity, so that the algorithmic advantage on quantum computers isn't overshadowed by the exponential speed-up from the fast access of quantum random memories. Moreover, we show that naive quantum matrix product estimation algorithm outperforms the best known classical matrix multiplication algorithm.

In section \ref{self}, we introduce quantum self-training, and point out its source of quantum supremacy. We give a concrete example, which is the quantum propagating nearest-neighbor algorithm. By comparing its time complexity to the classical case, we demonstrate its quantum supremacy. In section \ref{kmeans}, we present the quantum semi-supervised K-means clustering algorithm, and prove its quantum supremacy by time complexity analysis. 
We conclude the paper with outlooks in speeding up more complicated machine learning tasks including deep neural networks, along with an ultimate goal of using quantum-quantum learning to learn large quantum systems efficiently and reliably on a quantum computer.

\section{Framework of Quantum Semi-Supervised Learning}
\label{framework}

We first review semi-supervised learning, and then propose a general framework of quantum semi-supervised learning. Since we are going to introduce different types of quantum semi-supervised learning, we also call it classical or classical-classical semi-supervised learning. Semi-supervised learning lies between supervised and unsupervised learning: we have a small set of labeled data $\{ (\xsi, \zsi) \}_{i = 1}^l$, and a large set of unlabeled data $\{ \xsj \}_{j = l+1}^{l + u}$. 
Specifically, there are several different settings, including regression or classification with labeled and unlabeled data, constrained clustering, and dimensionality reduction with labeled instances whose reduced feature representation is given. We focus on the first setting. 

In supervised learning, the training sample is fully labeled, so the goal is always to label the future test data. However, in a semi-supervised setting, the training sample contains unlabeled data. Hence, there are two different goals in semi-supervised learning: inductive semi-supervised learning aims at predicting the labels of future test data, while transductive semi-supervised learning predicts the labels of the unlabeled instances in the training sample.

\begin{definition}
\textbf{Inductive Semi-Supervised Learning.}
Given a training sample 
$\{ (\xsi, \zsi) \}_{i = 1}^l$, $\{ \xsj \}_{j = l+1}^{l + u}$, inductive semi-supervised learning learns a function $f: \mathcal{X} \mapsto \mathcal{Y}$ so that $f$ is expected to be a good predictor on future data, beyond $\{ \xsj \}_{j = l+1}^{l + u}$.
\end{definition}

\begin{definition}
\textbf{Transductive Semi-Supervised Learning.}
Given a training sample $\{ (\xsi, \zsi) \}_{i = 1}^l, \{ \xsj \}_{j = l+1}^{l + u}$, transductive learning trains a function $f: \mathcal{X}^{l+u} \mapsto \mathcal{Y}^{l+u}$ so that $f$ is expected to be a good predictor on the unlabeled data $\{ \xsj \}_{j = l+1}^{l + u}$.
\end{definition}

There are different settings of quantum semi-supervised learning. First, classical-quantum semi-supervised learning encodes both labeled and unlabeled classical data in quantum states, and then uses quantum processors to carry out the learning phase. With a state-of-art design of quantum algorithm, classical-quantum semi-supervised learning executes a suitable learning task much faster than its classical-classical counterpart. 
Later, we will provide a general protocol in designing such algorithms, and show several examples. 

\begin{definition}
\textbf{Classical-Quantum Semi-Supervised Learning.}
Given a training sample $\{ (\xsi, \zsi) \}_{i = 1}^l, \{ \xsj \}_{j = l+1}^{l + u}$, classical-quantum semi-supervised learning encodes the data into quantum states, and then executes the learning task on a quantum computer. One convenient encoding is to map the training sample into product states: labeled data $\{ |\xsi\rangle |\zsi\rangle) \}_{i = 1}^l$, unlabeled data $\{ |\xsj\rangle \}_{j = l+1}^{l + u}$, stored in QRAM data structure. 
\end{definition}

Second, quantum-classical semi-supervised learning maps quantum data to a classical data structure, and then the problem turns into a classical-classical semi-supervised learning problem. While this sounds simple, there are many underlying subtleties. Suppose our training sample is $\{ (\rsi, \csi) \}_{i = 1}^l, \{ \rsj \}_{j = l+1}^{l + u}$, where $\rsi$'s are quantum states, and the labels $\csi$'s can be either classical or quantum. If we know $\rsi$'s and $\csi$'s exactly, then we essentially have classical data, and the mapping is completely trivial. Hence, the nontrivial case is when we have zero or partial knowledge of the training sample, but those quantum states are stored in a quantum memory nicely. We consider this case as the general setup of quantum-classical semi-supervised learning.

In this scenario, the novelty and challenge lie in the first step: how to efficiently and accurately extract classical information from the quantum data? Different mappings may result in different training performances and computational costs. The information loss when making quantum measurement, as well as the no-cloning theorem, adds additional difficulties to the problem. One simple but resource-consuming approach is to do efficient quantum tomography when many copies of same instances are present. To some extent, we have to learn the quantum states first. 

When certain limitations forbid us to do complete tomography, the problem becomes more interesting. It has more of an unsupervised learning flavor, as our training data, even the labeled ones, doesn't give out concrete knowledge. The way we process the quantum data may have crucial influence on the learning performance. Again, even in the regime of classical machine learning, pre-processing data can have significant impact on the training output. 
For inductive learning, this can be even trickier. When we have new quantum data coming in, are we going to process the test data the same way as we did for the training set? Intuitively, the answer is yes. A thorough discussion will be illuminated in future work.

\begin{definition}
\textbf{Quantum-Classical Semi-Supervised Learning.}
Given a training sample $\{ (\rsi, \csi) \}_{i = 1}^l, \{ \rsj \}_{j = l+1}^{l + u}$, where $\rsi$'s and $\csi$'s are partially known or completely unknown, quantum-classical semi-supervised learning extracts classical information from the initial quantum data using certain quantum channel, and then train the resulting classical data on a classical computer. For inductive learning, it uses the same channel to process the test quantum data, and then makes prediction on the corresponding classical data.
\end{definition}

Quantum-classical semi-supervised learning turns a complicated quantum problem to a simpler classical problem, and then solves it automatically on a classical computer, which is better understood at the current stage. 
The drawback is the loss of fidelity of the original data. To combat this, it is natural to skip the first step, which is the conversion of quantum data to classical data. When doing so, we learn the pattern of quantum data on a quantum computer, which is quantum-quantum learning.

\begin{definition}
\textbf{Quantum-Quantum Semi-Supervised Learning.}
Given a training sample $\{ (\rsi, \csi) \}_{i = 1}^l, \{ \rsj \}_{j = l+1}^{l + u}$, where $\rsi$'s and $\csi$'s are partially known or completely unknown, quantum-quantum semi-supervised learning executes the learning task on a quantum computer. 
\end{definition}

Intuitively, it is most natural to learn a quantum system on a quantum computer. However, with little classical information in this setting, we need a new set of theories and algorithms \cite{shangnan2019complexity, shangnan2021max, shangnan2021data} in the training process, even for things as simple as gradient descent \cite{shangnan2021qml}. Another challenge is that we are still limited by near-term quantum devices. Hence, right now we still need quantum-classical learning to aid our process of learning a quantum system.
In the near future, when fault-tolerant quantum computers are in commercial use, the default choice is to use quantum-quantum learning. 
These settings are not completely distinctive. Running tasks in a hybrid way can improve efficiency in time and space.

\section{Quantum Supremacy of Classical-Quantum Learning}
\label{supremacy}

In this section, we provide a generic protocol in designing classical-quantum learning algorithms that possess quantum supremacy. The idea is to take advantage of the fact that certain data acquisitions and manipulations can be done faster on a quantum computer \cite{QRAM, lloyd2013quantum, kerenidis2018q, hhl}.

\begin{theorem}
\textbf{QRAM data structure \cite{QRAM}.} 
Let $V \in \mathcal{R}^{N \times d}$, there exists a data structure to store the rows of $V$ such that

1. The time to insert, update, or delete a single entry $v_{ij}$ is 
$O\big(\log (N d) \big)$.

2. A quantum algorithm with access to the data structure can perform the following unitaries in time $O\big(\log (N d) \big)$.

(a) $|i\rangle |0\rangle \mapsto |i\rangle |v_i\rangle$ for $i \in [N]$.

(b) $|0\rangle \mapsto \sum_{i \in [N]} |v_i| |i\rangle$.
\end{theorem}

Using a classical RAM data structure, these tasks take $O(N d)$ time to complete. Hence, QRAM data structure provides an exponential speed-up. Many quantum algorithm references combine the memory access time complexity and algorithmic time complexity together when doing time complexity analysis. However, most traditional algorithm analysis takes the memory access time complexity as $O(1)$, because modern computers allow processor caches, memory level parallelism, etc. To put the comparison of quantum and classical algorithms on an equal footing, in this paper, we take memory access time complexity as a constant. It is good to keep in mind that quantum computers are inherently exponentially faster in reading and writing at a memory location.

\begin{theorem}
\textbf{Distance Estimation \cite{lloyd2013quantum, kerenidis2018q}.} 
Given data matrices $X \in \mathbb{R}^{l \times d}$ and $Y \in \mathbb{R}^{u \times d}$ stored in the QRAM data structure, where $\xsi$ is the $i$-th row of $X$, and $\ysj$ is the $j$-th row of $Y$. Suppose that the following unitaries $|i\rangle |0\rangle \mapsto |i\rangle |\xsi\rangle$, and $|j\rangle |0\rangle \mapsto |j\rangle |\ysj\rangle$ can be performed in time $\Lambda$ and the norms of the vectors are known. For any $\Delta > 0$ and $\epsilon > 0$, there exists a quantum algorithm that computes the $L^2$ distance between two vectors $\xsi$ and $\ysj$:
$|i\rangle |j\rangle |0\rangle \mapsto |i\rangle |j\rangle |\overline{d^2(\xsi, \ysj)} \rangle$, where $|\overline{d^2(\xsi, \ysj)} - d^2(\xsi, \ysj)| \leq \epsilon$ with probability at least $1 - 2 \Delta$ in time $T = \tilde{O}\Big(|\xsi||\ysj| \Lambda \log(1/\Delta) / \epsilon \Big)$.
\end{theorem}

\begin{theorem}\label{inner}
\textbf{Inner Product Estimation \cite{QRAM, kerenidis2018q}.}
Given data matrices $X \in \mathbb{R}^{l \times d}$ and $Y \in \mathbb{R}^{u \times d}$ stored in the QRAM data structure, where $\xsi$ is the $i$-th row of $X$, and $\ysj$ is the $j$-th row of $Y$. 
Suppose that the following unitaries $|i\rangle |0\rangle \mapsto |i\rangle |\xsi\rangle$, and $|j\rangle |0\rangle \mapsto |j\rangle |\ysj\rangle$ can be performed in time $\Lambda$ and the norms of the vectors are known. For any $\Delta > 0$ and $\epsilon > 0$, there exists a quantum algorithm that computes the inner product between two vectors $\xsi$ and $\ysj$:
$|i\rangle |j\rangle |0\rangle \mapsto |i\rangle |j\rangle |\overline{(\xsi, \ysj)}$, where $|\overline{(\xsi, \ysj)} - (\xsi, \ysj)| \leq \epsilon$ with probability at least $1 - 2 \Delta$ in time $T = \tilde{O}\Big(|\xsi||\ysj| \Lambda \log(1/\Delta) / \epsilon \Big)$.
\end{theorem}

The above results show that for quantum distance and inner product estimations, the algorithmic time complexity is $O(1)$ in terms of the dimensions of vectors, while the same classical calculation takes $O(d)$. 

Since matrix multiplications can be interpreted as calculating many inner products, we propose and prove the following theorem:

\begin{theorem}
\textbf{Matrix Product Estimation.}
Given data matrices $X \in \mathbb{R}^{l \times d}$ and $Y \in \mathbb{R}^{u \times d}$ stored in the QRAM data structure, where $\xsi$ is the $i$-th row of $X$, and $\ysj$ is the $j$-th row of $Y$. 
Suppose that the following unitaries $|i\rangle |0\rangle \mapsto |i\rangle |\xsi\rangle$, and $|j\rangle |0\rangle \mapsto |j\rangle |\ysj\rangle$ can be performed in time $\Lambda$ and the norms of the vectors are known. For any $\Delta > 0$ and $\epsilon > 0$, there exists a quantum algorithm that computes the product between two matrices $X$ and $Y^T$: $Z = X Y^T$, where $|\overline{z_{i j}} - z_{i j}| \leq \epsilon$ with probability at least $1 - 2 \Delta$ in time $T = \tilde{O}\Big(|\xsi||\ysj| l u \Lambda \log(1/\Delta) / \epsilon \Big)$.
\end{theorem}

\begin{proof} 
$z_{i j} = (\xsi, \ysj)$, which is the inner product of $\xsi$ and $\ysj$. To calculate matrix $Z$, we need to estimate $l u$ such inner products. By Theorem \ref{inner}, it takes time  $T = \tilde{O}\Big(|\xsi||\ysj| l u \Lambda \log(1/\Delta) / \epsilon \Big)$ to achieve $|\overline{z_{i j}} - z_{i j}| \leq \epsilon$ with probability at least $1 - 2 \Delta$.
\end{proof}

For naive matrix multiplications between an $m \times k$ matrix and an $k \times n$ matrix, quantum estimation takes time $O(m n)$, while classical calculation takes time $O(m n k)$. When dealing with $n \times n$ matrices, then the naive quantum matrix multiplication algorithm takes time $O(n^2)$. 
This is dramatic, because the classical matrix multiplication algorithm with best asymptotic complexity runs in $O(n^{2.3728596})$ time, and the naive classical algorithm runs in $O(n^3)$ time.

\begin{theorem}
\textbf{HHL Algorithm \cite{hhl}.}
Given a sparse $N \times N$ matrix $A$ with condition number $\kappa$, and a vector $b$. Suppose $M$ is a matrix, and $x$ is a vector such that $A x = b$. HHL algorithm estimates $x^{\dagger} M x$ in $\tilde{O} \big(\text{poly}(\log N, \kappa) \big)$ time.

In contrast, the classical algorithm that estimates $x^{\dagger} M x$ takes time $\tilde{O}(N \sqrt{\kappa})$. For the same task, HHL algorithm presents an exponential speed-up with respect to matrix size $N$.
\end{theorem}

Most machine learning algorithms require calculations of distance, inner product, matrix product and inverse. For these algorithms, if we can perform these calculations on a quantum computer, and make sure the quantum state preparation and classical information retrieval are not exponentially costly, we can realize quantum supremacy in all these algorithms. This sheds lights on training deep quantum neural networks faster, considering that many inner product calculations and matrix inverses are carried out in the training process.

For the remainder of the paper, we showcase general classes of quantum semi-supervised learning algorithms, and give some concrete examples. Moreover, we compare the time complexity of the classical and quantum versions, and demonstrate quantum supremacy in these scenarios.

\section{Quantum Self-Training}
\label{self}

Self-training is characterized by the fact that the learning process uses its own predictions to teach itself. It is also called self-teaching or bootstrapping because of this. Self-training assumes that its own predictions, at least the high confidences ones, tend to be correct. For classification tasks with well-separated clusters, this is usually the case.

The major goal of self-training is to learn an appropriate predictor $f$. We now show how this is done in quantum self-training, and point out which parts of the general quantum self-training algorithm possess quantum supremacy.

\begin{algorithm}
\textbf{Quantum Self-training}

Input: labeled data $\{ |\xsi\rangle |\zsi\rangle) \}_{i = 1}^l$, unlabeled data $\{ |\xsj\rangle \}_{j = l+1}^{l + u}$, stored in QRAM data structure.

1. Initially, let $L = \{ |\xsi\rangle |\zsi\rangle) \}_{i = 1}^l$, and $U = \{ |\xsj\rangle \}_{j = l+1}^{l + u}$.

2. Repeat until convergence:

3. Train $f$ from $L$ using supervised learning.

4. Apply $f$ to the unlabeled instances in $U$.

\end{algorithm}

In line 3, we can apply existing quantum supervised learning algorithms \cite{lloyd2013quantum, QMLNature, KernelPRL} or derive new quantum supervised learning algorithms based on principles mentioned in section \ref{supremacy}, so that quantum supremacy is achieved.
In line 4, because QRAM reads and writes data exponentially faster than RAM, we naturally have quantum supremacy.

We now give a simple and concrete example of quantum self-training, which is the propagating nearest-neighbor algorithm. We first show the classical version and then propose the quantum version.

\begin{algorithm}
\textbf{Classical Propagating Nearest-Neighbor Classifier.}

Input: labeled data $\big\{ \big(\xsi, f(\xsi)\big) \big\}_{i = 1}^l$, unlabeled data $\{ \xsj \}_{j = l+1}^{l + u}$, distance function $d()$.

1. Initially, let $L = \big\{ \big(\xsi, f(\xsi)\big) \big\}_{i = 1}^l $, and $U = \{ \xsj \}_{j = l+1}^{l + u}$.

2. Repeat until $U$ is empty:

3. Select $x = \arg \min_{x \in U} \min_{y \in L} d(x, y)$, $y = \arg \min_{y \in L} \min_{x \in U} d(x,y)$.

4. Set $f(x) = f(y)$, where $f(y)$ is the label of $y$. Break ties randomly.

5. Remove $x$ from $U$. Add $(x, f(x))$ to $L$.

\end{algorithm}

In each iteration, the classical algorithm selects the unlabeled instance that is closest to any "labeled" instance, i.e., any instance currently in $L$. These "labeled" instances may from the original labeled set, or were labeled in previous iterations. The selected instance is then assigned the label of its nearest neighbor and inserted into $L$ as if it were truly labeled data. The process repeats until all instances have been added to $L$.

\begin{algorithm}
\textbf{Quantum Propagating Nearest-Neighbor Classifier.}

Input: labeled data $\{ |\xsi\rangle |\zsi\rangle \}_{i = 1}^l$, unlabeled data $\{ |\xsj\rangle \}_{j = l+1}^{l + u}$, stored in QRAM data structure. Distance function $d()$.

1. Initially, let $L = \{ |\xsi\rangle |\zsi\rangle \}_{i = 1}^l $, and $U = \{ |\xsj\rangle \}_{j = l+1}^{l + u}$. 

2. Repeat until $U$ is empty:

3. Step 1: Point Distance Estimation.

Perform the map:

\begin{equation*}
   \otimes_{|\xsi\rangle \in L} |i\rangle \otimes_{|\xsj\rangle \in U} |j\rangle |0\rangle
   \mapsto \otimes_{|\xsi\rangle \in L} |i\rangle \otimes_{|\xsj\rangle \in U} |j\rangle |\overline{d^2(\xsi, \xsj)}\rangle.
\end{equation*}

4. Step 2: Distance Minimization.

Find the minimum distance among $\{\overline{d^2(\xsi, \xsj)}\}_{|\xsi\rangle \in L, |\xsj\rangle \in U}$, record the corresponding register number $|i\rangle$ and $|j\rangle$.

5. Step 3: Label Assignment.

Set $|\zsj\rangle = |\zsi\rangle $ to be the label of $|\zsj\rangle$.

Remove $|\xsj\rangle$ from $U$, add $|\xsj\rangle |\zsj\rangle $ to $L$.

\end{algorithm}

In the quantum algorithm, we follow the same logic as its classical version, but store and manipulate data on a quantum computer. At each iteration, let $l = |L|, u = |U|$, the point distance estimation takes time $O (l u)$; the distance minimization takes $O (l u)$, and the label assignment takes $O(1)$, so the combined time complexity is $O (l u)$.

As a comparison, at each iteration, the corresponding classical algorithm takes time $O(l u d)$ to calculate the distance, and then $O (l u)$ for distance minimization, and $O(1)$ for label assignment, with a combined time complexity $O( l u d)$. 

The above discussion demonstrates the quantum supremacy of the quantum propagating nearest neighbor classifier. In particular, when the data points are of high dimension, the quantum speed-up becomes significant. By creating state-of-art superpositions, it is possible to reduce the running time of quantum distance estimation even further. Last but not least, we should keep in mind that QRAM processes data exponentially faster than RAM.

\section{Quantum Semi-Supervised K-Means}
\label{kmeans}

$K$-means clustering is a simple and popular unsupervised machine learning algorithm. It plays a significant role in cluster analysis. There are different ways of extending it to a semi-supervised setting. We consider the case when some data points are labeled, i.e. we have labeled data $\{ (\xsi, \zsi) \}_{i = 1}^l$, and unlabeled data $\{ \xsj \}_{j = l+1}^{l + u}$. One approach to classical semi-supervised K-means clustering is the following algorithm.

\begin{algorithm}
\textbf{Classical Semi-Supervised K-Means.}

Input: labeled data $\{ \ysi = (\xsi, \zsi) \}_{i = 1}^l$, unlabeled data $\{ \xsj \}_{j = l+1}^{l + u}$, distance function $d()$, number of centroids $k$. The total number of data points is $N = l + u$, with $l << u$.

1. Initially, for any $m \in [k]$, denote the set of labeled data points whose label is $m$ as $S_m$. If $S_m = \varnothing$, then select centroid $c_m^0$ randomly. If $S_m \neq \varnothing$, then $c_m^0 = \frac{1}{|S_m|}\sum_{\ysi \in S_m} \xsi$.

2. $t = 0.$

3. Repeat until convergence:

4. Step 1: Centroid Distance Calculation.

Calculate $d(\xsj, c_m^t)$ for all $j \in [l+1, N]$, and all $m \in [k]$.

5. Step 2: Cluster Assignment.

For labeled data, assign their original label; for unlabeled data, find the minimum distance among $\{d(\xsj, c_m^t) \}$, and assign $m$ as the label of $\xsj$, i.e. $\zsj = m$.

6. Step 3: Centroid Update.

For each $m$, let 

\begin{equation}
    c_m^{t+1} = \frac{\sum_{j=1}^{N} 1\{\zsj = m\} \xsi}{\sum_{j=1}^{N} 1\{\zsj = m\}}.
\end{equation}

7. $t=t+1$.

\end{algorithm}

We now propose the quantum semi-supervised $K$-means algorithm.

\begin{algorithm}
\textbf{Quantum Semi-Supervised K-Means.}

Input: labeled data $\{ \ysi = |\xsi\rangle |\zsi\rangle \}_{i = 1}^l$, unlabeled data $\{ |\xsj\rangle \}_{j = l+1}^{l + u}$, stored in QRAM data structure. Number of centroids $k$. Distance function $d()$. 
The total number of data points is $N = l + u$, with $l < < u$. The dimension of each data vector $|\xsi\rangle$ is $d$.

1. Initially, for any $m \in [k]$, denote the set of labeled data points whose label is $|m\rangle$ as $S_m$. If $S_m = \varnothing$, then select centroid $|c_m^0\rangle$ randomly. If $S_m \neq \varnothing$, then $|c_m^0\rangle = \frac{1}{|S_m|}\sum_{\ysi \in S_m} |\xsi\rangle$.

2. $t = 0.$

3. Repeat until convergence:

4. Step 1: Centroid Distance Estimation.

With states $|\xsj\rangle$'s and $|c_m^t\rangle$'s, perform the map:

\begin{equation*}
    \frac{1}{\sqrt{N}} \sum_{j=1}^{N} |j\rangle \otimes_{m \in [k]} |m\rangle |0\rangle
    \mapsto \frac{1}{\sqrt{N}} \sum_{j=1}^{N} |j\rangle \otimes_{m \in [k]} |m\rangle 
    |\overline{d^2(\xsj, c_m^t)}\rangle.
\end{equation*}

5. Step 2: Cluster Assignment.

For labeled data, assign their original label; for unlabeled data, find the minimum distance among $\{\overline{d(\xsj, c_m^t)} \}$, and assign $m$ as the label of $|\xsj\rangle$, i.e. $|\zsj\rangle = |m\rangle$. Next, uncompute Step 1 to create the superposition of all points and their labels

\begin{equation*}
     \frac{1}{\sqrt{N}} \sum_{j=1}^{N} |j\rangle \otimes_{m \in [k]} |m\rangle 
    |\overline{d^2(\xsj, c_m^t)}\rangle \mapsto
    \frac{1}{\sqrt{N}} \sum_{j=1}^{N} |j\rangle |\zsj\rangle.
\end{equation*}

6. Step 3: Label Measurement.

Denote the set of data points whose label is $|m\rangle$ as $S_m^t$. Measure the label register to obtain a state $|\chi_m^t\rangle = \frac{1}{|S_m^t|} \sum_{\ysj \in S_m^t} |j\rangle$, with probability $\frac{|S_m^t|}{N}$.

7. Step 4: Centroid Update.

Denote the data matrix as $V \in \mathbb{R}^{N \times d}$. Perform matrix multiplication to obtain

\begin{equation*}
    |c_m^{t+1}\rangle = V^T |\chi_m^t\rangle =  \frac{1}{|S_m^t|} \sum_{\ysj \in S_m^t} |\xsj\rangle.
\end{equation*}

8. $t = t+1.$

At each iteration, the quantum distance estimation takes time $O(k)$, while its classical counterpart takes time $O(N k d)$; the quantum cluster assignment takes time $O(k)$, while its classical counterpart takes time $O(N k)$. In the quantum case, the time complexity of label measurement is considered as $O(1)$, and then the quantum matrix multiplication of centroid update takes time $O(N)$ for each centroid. Hence, steps 3 and 4 take time $O(N k)$ combined. At the same time, the classical centroid update process takes time $O(N k d)$. 
Hence, the quantum algorithm is faster at every stage of the training, which demonstrates its supremacy.

\end{algorithm}

%\section{Quantum Semi-Supervised EM}

\section{Conclusion}

We propose the framework of quantum semi-supervised learning, which resolves two long-standing challenges of machine learning at the same time. Semi-supervised learning tackles issue of the lack of labeled data, and quantum computation provides dramatic speed-ups so that limit of computational power is no longer the issue. We provide a protocol that systematically designs quantum machine learning algorithms with quantum supremacy and showcase examples. The recipe can be extended beyond supervised, unsupervised, and semi-supervised learning. In the future, we will demonstrate how we provide quantum speed-ups to deep neural networks \cite{shangnan2021d}.  
We also aim at developing more complicated quantum semi-supervised learning algorithms, for example, quantum co-training, quantum graph-based training \cite{shangnan2021s}. 
Furthermore, we keep an ultimate goal in mind: use quantum-quantum learning to learn large quantum systems efficiently and reliably on a quantum computer \cite{shangnan2021qml}.

\section*{Acknowledgement}
We thank Xinyu Ren for suggesting a useful reference.
Z.S. is supported by the Simons Foundation and the National Science Foundation, under Grant No. 2111998.

\bibliographystyle{unsrt}
\bibliography{ref}

\end{document}